\documentclass{article} 
\usepackage{nips14submit_e,times}
\usepackage{hyperref}
\usepackage{url}
\usepackage{amsmath}
\usepackage{amsthm}
\usepackage{amssymb}
\usepackage{algorithm}
\usepackage{algpseudocode}
\usepackage{graphicx}
\usepackage{picins}

\DeclareMathOperator\argmin{arg\,min}
\DeclareMathOperator\argmax{arg\,max}

\newtheorem{thm}{Theorem}
\newtheorem{lem}{Lemma}
\newtheorem{defi}{Definition}
\newtheorem{prop}{Proposition}

\title{On Unconstrained Quasi-Submodular Function Optimization}

\author{
Jincheng Mei, Kang Zhao and Bao-Liang Lu\\
Center for Brain-Like Computing and Machine Intelligence\\
Department of Computer Science and Engineering\\
Key Laboratory of Shanghai Education Commission for Intelligent Interaction and Cognitive Engineering\\
Shanghai Jiao Tong University\\
800 Dong Chuan Road, Shanghai 200240, China\\
\texttt{\{jcmei,sjtuzk,bllu\}@sjtu.edu.cn} \\
}

\nipsfinalcopy 

\begin{document}

\maketitle

\begin{abstract}
With the extensive application of submodularity, its generalizations are constantly being proposed. However, most of them are tailored for special problems. In this paper, we focus on quasi-submodularity, a universal generalization, which satisfies weaker properties than submodularity but still enjoys favorable performance in optimization. Similar to the diminishing return property of submodularity, we first define a corresponding property called the {\em single sub-crossing}, then we propose two algorithms for unconstrained quasi-submodular function minimization and maximization, respectively. The proposed algorithms return the reduced lattices in $\mathcal{O}(n)$ iterations, and guarantee the objective function values are strictly monotonically increased or decreased after each iteration. Moreover, any local and global optima are definitely contained in the reduced lattices. Experimental results verify the effectiveness and efficiency of the proposed algorithms on lattice reduction.
\end{abstract}

\section{Introduction}

Given a ground set $N=\{1,2,\cdots,n\}$, a set function $F: 2^N \mapsto \mathbb{R}$ is said to be submodular \cite{fujishige2005submodular} if $\forall X,Y \subseteq N$,
\begin{equation*}
    F(X)+F(Y) \ge F(X\cap{Y})+F(X\cup{Y}).
\end{equation*}
An equivalent definition is given as following, \emph{i.e.}, $\forall A \subseteq B \subseteq N$, $i \in N \setminus B$,
\begin{equation*}
    F(i|A) \ge F(i|B),
\end{equation*}
where $F(i|A) \triangleq F(A + i)-F(A)$ is called the marginal gain of $i$ with respect to $A$. It implies that submodular functions capture the {\em diminishing return} property. To facilitate our presentation, we use $F(A+i)$ to refer to $F(A\cup\{i\})$, and $F(A-i)$ to refer to $F(A \setminus\{i\})$.

Submodularity is widely applied in economics, combinatorics, and machine learning, such as welfare allocation \cite{vondrak2008optimal},
sensor placement \cite{krause2008near}, feature selection \cite{das2011submodular}, and computer vision \cite{liu2011entropy}, to name but a few.

With the wide application of submodularity, it has many generalizations. For example, Singh et al. \cite{singh2012bisubmodular} formulate multiple sensor placement and multimodal feature selection as bisubmodular function maximization, where the objectives have multiple set arguments. Golovin and Krause \cite{golovin2011adaptive} introduce the concept of adaptive submodularity to make a sequence of adaptive decisions with uncertain responses. Feige \cite{feige2009maximizing} proposes maximizing subadditive functions on welfare problems to capture the complement free property of the utility functions. However, all the mentioned generalizations of submodularity enjoy benefits in special application scenarios (multiset selection, adaptive decision, and complement free allocation).

In this paper, we study a universal generalization. Submodularity is often viewed as the discrete analogue of convexity \cite{lovasz1983submodular}. One of the most important generalizations of convexity is quasi-convexity \cite{boyd2004convex}. Quasi-convex functions satisfy some weaker properties, but still benefit much from the optimization perspective. More specifically, quasi-convex constraints can be easily transformed to convex constraints via sublevel sets, and quasi-convex optimization problems can be solved through a series of convex feasibility problems using bisection methods \cite{boyd2004convex}. Considering the celebrated analogue between submodularity and convexity, a natural question is whether submodularity has similar generalizations which satisfy weaker properties but still enjoy favorable performance in optimization? In this paper, we positively answer this question and refer to this generalization as quasi-submodularity.

As aforementioned, quasi-submodularity is a weaker property than submodularity. Similar to the diminishing return property of submodular functions, we first define a corresponding property called single sub-crossing. Then we propose two algorithms for unconstrained quasi-submodular minimization and maximization, respectively. Our theoretical analyses show that the proposed algorithms strictly increase or decrease the objective function values after each iteration. The output reduced lattices can be obtained in $\mathcal{O}(n)$ iterations, which contain all the local and global optima of the optimization problems. The theoretical and experimental results indicate that although quasi-submodularity is a weaker property than submodularity, it enjoys favorable performance in optimization.

The rest of the paper is organized as follows. In Section $2$, we introduce the concept of quasi-submodularity and define the single sub-crossing property. In Section $3$ and Section $4$, we present the efficient algorithms and theoretical analyses for unconstrained quasi-submodular function minimization and maximization, respectively. After that, we provide some discussion in Section $5$. Experimental results in Section $6$ verify the effectiveness of the proposed algorithms on lattice reduction. Finally, we introduce some related work in Section $7$ and give some conclusions about our work in Section $8$.

\section{Quasi-Submodularity}

It is well known that the term semi-modular is taken from lattice theory \cite{edmonds2003submodular}. A lattice is a partially ordered set, which contains the supremum and infimum of each element pair. Here, we introduce a very useful lattice.

\begin{defi}[Set Interval Lattice]
Given two ground sets $A$, $B$, a set interval lattice $\mathcal{L} = [A,B]$ is defined as $\{U \ |\ A \subseteq U \subseteq B\}$. $\mathcal{L}$ is not empty if and only if $A \subseteq B$.
\end{defi}
In the set interval lattice, the partially order relation is defined as the set inclusion $\subseteq$. A set $S \in \mathcal{L}$ iff $A \subseteq S \subseteq B$. Obviously, $\forall X, \ Y \in \mathcal{L}$, we have $X \cap Y, \ X \cup Y \in \mathcal{L}$, thus $\mathcal{L}$ is a lattice.

The concept of quasi-supermodularity is first proposed by Milgrom and Shannon \cite{milgrom1994monotone} in economic fields . Quasi-supermodularity captures the monotonicity of the solutions as the problem parameters change, and has been proved useful in game theory \cite{leininger2006fending}, parametric cuts \cite{granot2012structural}, and discrete convex analysis \cite{murota2003discrete}. Following \cite{milgrom1994monotone}, we give the definition of quasi-submodularity.

\begin{defi}[QSB]
A set function $F: 2^N \mapsto \mathbb{R}$ is quasi-submodular function if $\forall X,Y \subseteq N$, {\bf both} of the following conditions are satisfied
\begin{align*}\tag{1}
\begin{split}
    F(X\cap{Y}) &\ge F(X) \Rightarrow F(Y) \ge F(X\cup{Y}),\\
    F(X\cap{Y}) &> {F(X)} \Rightarrow F(Y) > F(X\cup{Y}).
\end{split}
\label{eq:1}
\end{align*}
\end{defi}
The following proposition implies that the concept of quasi-submodularity is a generalization of submodularity.
\begin{prop}
Any submodular function is quasi-submodular function, but {\bf not} vice versa.
\end{prop}
\begin{proof}
Suppose $F: 2^N \mapsto \mathbb{R}$ is a submodular function, and $F$ is not a quasi-submodular function. Then we have $F(X \cap Y) \ge F(X)$, $F(Y) < F(X \cup Y)$, or $F(X \cap Y) > F(X)$, $F(Y) \le F(X \cup Y)$. Both of the two cases lead to $F(X)+F(Y) < F(X\cap{Y})+F(X\cup{Y})$, which contradicts the definition of submodularity.

A counterexample is given to prove a quasi-submodular function may not be a submodular function. Suppose $N=\{1,2\}$, $F(\emptyset) = 1$, $F(\{1\}) = 0$, $F(\{2\}) = 1.5$, and $F(\{1,2\}) = 1$. It is easy to check that $F$ satisfies the definition of QSB. But $F$ is not a submodular function, since $F(\{1\})+F(\{2\}) < F(\emptyset)+F(\{1,2\})$. Actually, $F$ is a supermodular function.
\end{proof}
Similar to the diminishing return property of submodular functions, we define a corresponding property for quasi-submodularity, and name it as {\em single sub-crossing}.

\begin{defi}[SSBC]
A set function $F: 2^N \mapsto \mathbb{R}$ satisfies the single sub-crossing property if $\forall A \subseteq B \subseteq N,\ i \in N \setminus B$, {\bf both} of the following conditions are satisfied
\begin{align*}\tag{2}
\begin{split}
    F(A) &\ge F(B) \Rightarrow F(A+i) \ge F(B+i),\\
    F(A) &> F(B) \Rightarrow F(A+i) > F(B+i).
\end{split}
\label{eq:2}
\end{align*}
\end{defi}
As mentioned before, submodularity and diminishing return property are equivalent definitions. Analogously, quasi-submodularity and single sub-crossing property  are also equivalent.
\begin{prop}
Any quasi-submodular function satisfies the single sub-crossing property, and vice versa.
\end{prop}
\begin{proof}
Suppose $F: 2^N \mapsto \mathbb{R}$ is a quasi-submodular function. $\forall A \subseteq B \subseteq N,\ i \in N \setminus B$, let $X=B,\ Y=A+i$ in (\ref{eq:1}). It is obvious that $F$ satisfies the SSBC property.

On the other hand, suppose $F$ satisfies the SSBC property. $\forall X,Y \subseteq N$, we denote $Y \setminus X = \{i_1,i_2,\cdots,i_k\}$. Based on the SSBC property, if $F(X\cap{Y}) \ge (>) F(X)$, then we have $F(X\cap{Y} + i_1) \ge (>) F(X+i_1)$. Similarly, we have $F(X\cap{Y} + i_1 + i_2) \ge (>) F(X+i_1+i_2)$. Repeating the operation until $i_k$ is added, we get $F(Y) \ge (>) F(X\cup{Y})$.
\end{proof}
Note that in the proof above, if we exchange $X$ and $Y$, \emph{i.e.}, let $X=A+i,\ Y=B$, we will get
\begin{align*}
    F(A) &\ge F(A+i) \Rightarrow F(B) \ge F(B+i),\\
    F(A) &> F(A+i) \Rightarrow F(B) > F(B+i).
\end{align*}
We can rewrite it using the marginal gain notation, \emph{i.e.}, $\forall A \subseteq B \subseteq N,\ i \in N \setminus B$,
\begin{equation*}\tag{3}
    F(i|A) \le (<)\ 0 \Rightarrow F(i|B) \le (<)\ 0.
\label{eq:3}
\end{equation*}
Note that although $X$ and $Y$ are symmetric and interchangeable in (\ref{eq:1}), we get a representation which is different with the SSBC property. Actually, (\ref{eq:3}) is a weaker condition than (\ref{eq:1}). The proposed algorithms work on the weaker notion (\ref{eq:3}), and the results also hold for quasi-submodularity.

\section{Unconstrained Quasi-Submodular Function Minimization}

In this section, we are concerned with general unconstrained quasi-submodular minimization problems, where the objective functions are given in the form of value oracle. Generally, we do not make any additional assumptions (such as nonnegative, monotone, symmetric, etc) except quasi-submodularity.

Very recently, Iyer et al. \cite{iyer2013fast} propose a discrete Majorization-Minimization like submodular function minimization algorithm. In \cite{iyer2013fast}, for each submodular function, a tight modular upper bound is established at the current working set, then this bound is minimized as the surrogate function of the objective function. But for quasi-submodular function, there is no known superdifferential, and it can be verified that the upper bounds in \cite{iyer2013fast} are no longer bounds for quasi-submodular functions. Actually, without submodularity, quasi-submodularity is sufficient to perform lattice reduction. Consequently, we design the following algorithm.

\begin{algorithm}[h]
\caption{Unconstrained Quasi-Submodular Function Minimization (UQSFMin)}
\label{algo1}
\begin{algorithmic}[1]
\Require
Quasi-submodular function $F$, $N=\{1,2,...,n\}$, $X_0 \subseteq N$, $t \leftarrow 0$.
\Ensure
$X_t$ as a local optimum of $\min\limits_{X \subseteq N}{F(X)}$.
\State At Iteration $t$, find $U_t = \{u \in N \setminus X_t \ |\  F(u|X_t) < 0\}$. $Y_t \leftarrow X_t \cup U_t$.
\State Find $D_t = \{d \in X_t \ |\ F(d|Y_t - d) > 0\}$. $X_{t+1} \leftarrow Y_t \setminus D_t$.
\State If $X_{t+1} = X_t$ (iff $U_t = D_t = \emptyset$), stop and output $X_t$.
\State $t \leftarrow t+1$. Back to Step $1$.
\end{algorithmic}
\end{algorithm}

$X$ is a local minimum means $\forall i \in X$, $F(X-i) \ge F(X)$, and  $\forall j \in N \setminus X$, $F(X+j) \ge F(X)$.

Algorithm $1$ has several nice theoretical guarantees. First, the objective function values are strictly decreased after each iteration, as the following lemma states.
\begin{lem}
After each except the last iteration of Algorithm $1$, the objective function value of the working set is strictly monotonically decreased, i.e., $\forall t$, $F(X_{t+1}) < F(X_t)$.
\end{lem}
\begin{proof}
We prove $F(Y_t) < F(X_t)$. $F(X_{t+1}) < F(Y_t)$ can be proved using a similar approach. Suppose $U_t \not= \emptyset$. Define $U^k \in \argmin_{U \subseteq U_t: |U| = k}{F(X_t \cup U)}$, and $Y_t^k = X_t \cup U^k$. According to the algorithm, $\forall u \in U_t \setminus U^k$, $F(u|X_t) < 0$. Since $X_t \subseteq Y_t^k$, and $u \not \in Y_t^k$, based on the SSBC property, we have $F(u | Y_t^k) < 0$. This implies $F(Y_t^{k+1}) \le F(X_t \cup (U^k+u))  < F(X_t \cup U^k) = F(Y_t^k)$. Note that $F(Y_t^1) = \min_{u \in U_t}{F(X_t + u)} < F(X_t) = F(Y_t^0)$. We then have $F(Y_t) = F(Y_t^{|U_t|}) < F(Y_t^{|U_t|-1}) < \cdots < F(Y_t^0) = F(X_t)$.
\end{proof}
If we start from $X_0 = Q_0 \triangleq \emptyset$, after one iteration, we will get $X_1 = Q_1 = \{i \ |\  F(i|\emptyset) < 0\}$. Similarly, if we start from $X_0 = S_0 \triangleq N$, we will get $X_1 = S_1 = \{i \ |\  F(i|N-i) \le 0\}$. Based on the SSBC property, we have $\forall i \in N, F(i|\emptyset) < 0 \Rightarrow F(i|N-i) < 0$, \emph{i.e.}, $Q_1 \subseteq S_1$. Thus the reduced lattice $\mathcal{L} = [Q_1,S_1] \subseteq [\emptyset,N]$ is not empty, and we show that it contains all the global minima.
\begin{lem}
Any global minimum of $F(X)$ is contained in the lattice $\mathcal{L} = [Q_1,S_1]$, i.e., $\forall X_* \in \argmin_{X \subseteq N}{F(X)}$, $Q_1 \subseteq X_* \subseteq S_1$.
\end{lem}
\begin{proof}
We prove $Q_1 \subseteq X_*$. $X_* \subseteq S_1$ can be proved in a similar way. Suppose $Q_1 \not \subseteq X_*$, \emph{i.e.}, $\exists \ u \in Q_1$, $u \not \in X_*$. According to the definition of $Q_1$, $F(u|\emptyset) < 0 $. Since $\emptyset \subseteq X_*$, based on the SSBC property, we have $F(u|X_*) < 0$, which implies $F(X_*+u) < F(X_*)$. This contradicts the optimality of $X_*$.
\end{proof}

If we start Algorithm $1$ from $X_0 = Q_0 = \emptyset$, suppose we get $Q_t$ after $t$ iterations. It is easy to check that, due to the SSBC property, in each iteration, $Q_t$ only adds elements. So we get a chain $\emptyset = Q_0 \subseteq Q_1 \subseteq \cdots \subseteq Q_t \subseteq \cdots \subseteq Q_+$, where $Q_+$ is the final output when the algorithm terminates. Similarly, if we start from $X_0 = S_0 = N$, we can get another chain $S_+ \subseteq \cdots \subseteq S_t \subseteq \cdots \subseteq S_1 \subseteq S_0 = N$. We then prove that the endpoint sets of the two chains form a lattice, which contains all the local minima of $F$.

\begin{lem}
Any local minimum of $F(X)$ is contained in the lattice $\mathcal{L} = [Q_+,S_+]$.
\end{lem}
\begin{proof}
Let $P$ be a local minimum. In the proof of Lemma $2$, we use singleton elements to construct contradictions, so we have $Q_1 \subseteq P \subseteq S_1$. Suppose $Q_t \subseteq P \subseteq S_t$, we then prove $Q_{t+1} \subseteq P \subseteq S_{t+1}$. First, we suppose $Q_{t+1} \not \subseteq P$. Because $Q_{t+1} = Q_t \cup U_t$, $\exists \ u \in U_t$, $u \not \in P$. According to the definition of $U_t$, $F(u|Q_t) < 0$. Since $Q_t \subseteq P$, based on the SSBC property, we have $F(u|P) < 0$. This indicates $F(P+u) < F(P)$, which contradicts the local optimality of $P$. Hence $Q_{t+1} \subseteq P$. And $P \subseteq S_{t+1}$ can be proved in a similar way.
\end{proof}

Moreover, the two endpoint sets $Q_+$ and $S_+$ are local minima.

\begin{lem}
$Q_+$ and $S_+$ are local minima of $F(X)$.
\end{lem}
\begin{proof}
We prove for $Q_+$. The $S_+$ case is similar. According to the algorithm, $\forall i \in N \setminus Q_+$, $F(i|Q_+) \ge 0$. If $\exists \ j \in Q_+$, such that $F(j|Q_+ - j) > 0$, then we can suppose $j$ was added into $Q_+$ at a previous iteration $t$. Since $Q_t-j \subseteq Q_+ - j$, based on the SSBC property, we have $F(j|Q_t-j) > 0$. This contradicts the proof of Lemma $1$.
\end{proof}

Because a global minimum is also a local minimum, Lemma $3$ results in the following theorem.

\begin{thm}
Any global minimum of $F(X)$ is contained in the lattice $\mathcal{L} = [Q_+,S_+]$, i.e., $\forall X_* \in \argmin_{X \subseteq N}{F(X)}$, $Q_+ \subseteq X_* \subseteq S_+$.
\end{thm}

\section{Unconstrained Quasi-Submodular Function Maximization}

Unconstrained submodular function minimization problems can be exactly optimized in polynomial time \cite{orlin2009faster}. Yet unconstrained submodular maximization is NP-hard \cite{feige2011maximizing}. The best approximation ratio for unconstrained nonnegative submodular maximization is $1/2$ \cite{buchbinder2012tight}, which matches the known hardness result \cite{feige2011maximizing}. As a strict superset of submodular case, unconstrained quasi-submodular maximization is definitely NP-hard.

Iyer et al. \cite{iyer2013fast} also propose a discrete Minorization-Maximization like submodular maximization algorithm. They employ the permutation based subdifferential \cite{fujishige2005submodular} to construct tight modular lower bounds, and maximize the lower bounds as surrogate functions. With different permutation strategies, their algorithm actually mimics several existing approximation algorithms, which means their algorithm does not really reduce the lattices in optimization. In addition, for quasi-submodular cases, it also can be verified that the lower bounds in \cite{iyer2013fast} are no longer bounds, and quasi-submodular functions have no known subdifferential. Thus, even generalizing their algorithm is impossible.

We find Buchbinder et al. \cite{buchbinder2012tight} propose a simple linear time approximation method. The algorithm maintains two working sets, $S_1$ and $S_2$, and $S_1 \subseteq S_2$. At the start, $S_1 = \emptyset$ and $S_2 = N$. Then at each iteration, one element $i \in S_2 \setminus S_1$ is queried to compute its marginal gains over the two working sets, \emph{i.e.}, $F(i|S_1)$ and $F(i|S_2-i)$. If $F(i|S_1) + F(i|S_2-i) \ge 0$, then $S_1 \leftarrow S_1 + i$, otherwise $S_2 \leftarrow S_2 - i$. After $n$ iterations, the algorithm outputs $S_1 = S_2$. This algorithm is efficient and reaches an approximation ratio of $1/3$. However, the approximate algorithm may mistakenly remove a certain element $e \in X_*$ from $S_2$, or add an element $u \not \in X_*$ into $S_1$. Here, $X_*$ is referred to as a global maximum. Consequently, the working lattices of their algorithm may not contain the global optima.

By contrast, we want to reduce the lattices after each iteration while avoid taking erroneous steps. Fortunately, we find that if we simultaneously maintain two working sets at each iteration, and take steps in a "crossover" method, quasi-submodularity can provide theoretical guarantees that the output lattices definitely contain all the global maxima. Hence, we propose the following algorithm.

\begin{algorithm}[h]
\caption{Unconstrained Quasi-Submodular Function Maximization (UQSFMax)}
\label{algo2}
\begin{algorithmic}[1]
\Require
Quasi-submodular function $F$, $N=\{1,2,...,n\}$, $X_0 \leftarrow \emptyset$, $Y_0 \leftarrow N$, $t \leftarrow 0$.
\Ensure
Lattice $[X_t,Y_t]$.
\State At Iteration $t$, find $U_t = \{u \in Y_t \setminus X_t \ |\  F(u|Y_t-u) > 0\}$. $X_{t+1} \leftarrow X_t \cup U_t$.
\State Find $D_t = \{d \in Y_t \setminus X_t \ |\ F(d|X_t) < 0\}$. $Y_{t+1} \leftarrow Y_t \setminus D_t$.
\State If $X_{t+1} = X_t$ and $Y_{t+1} = Y_t$, stop and output $[X_t,Y_t]$.
\State $t \leftarrow t+1$. Back to Step $1$.
\end{algorithmic}
\end{algorithm}

To ensure the result lattice is not empty, we prove that after each iteration Algorithm $2$ maintains a nonempty lattice as the following lemma shows.
\begin{lem}
At each iteration of Algorithm $2$, the lattice $[X_t,Y_t]$ is not empty, i.e., $\forall t$, $X_t \subseteq Y_t$.
\end{lem}
\begin{proof}
According to the definition, we have $X_0 \subseteq Y_0$. Suppose $X_t \subseteq Y_t$, we then prove $X_{t+1} \subseteq Y_{t+1}$. Because $U_t, \ D_t \subseteq Y_t \setminus X_t$, if we prove $U_t \cap D_t = \emptyset$, $X_{t+1} \subseteq Y_{t+1}$ will be satisfied. According to the algorithm, $\forall u \in U_t$, $F(u|Y_t-u) >0$. Since $X_t \subseteq Y_t-u$, and $u \not \in Y_t-u$, based on the SSBC property, we have $F(u|X_t) > 0$, which implies $u \not \in D_t$.
\end{proof}
Algorithm $2$ also has several very favorable theoretical guarantees. First, the objective function values are strictly increased after each iteration, as the following lemma states.
\begin{lem}
After each except the last iteration of Algorithm $2$, the objective function values of endpoint sets of lattice $[X_t,Y_t]$ are strictly monotonically increased, i.e., $\forall t$, $F(X_{t+1}) > F(X_t)$ or $F(Y_{t+1}) > F(Y_t)$.
\end{lem}
\begin{proof}
We prove $F(X_{t+1}) > F(X_t)$. $F(Y_{t+1}) > F(Y_t)$ can be proved using a similar approach. Suppose $U_t \not= \emptyset$. Define $U^k \in \argmax_{U \subseteq U_t: |U| = k}{F(X_t \cup U)}$, and $X_t^k = X_t \cup U^k$. According to the algorithm, $\forall u \in U_t \setminus U^k$, $F(u|Y_t-u) > 0$. Since $X_t^k \subseteq Y_t-u$, and $u \not \in Y_t-u$, based on the SSBC property, we have $F(u|X_t^k) > 0$. This indicates $F(X_t^{k+1}) \ge F(X_t \cup (U^k+u)) > F(X_t \cup U^k) = F(X_t^k)$. Note that $F(X_t^1) = \max_{u \in U_t}{F(X_t + u)} > F(X_t) = F(X_t^0)$. We then have $F(X_{t+1}) = F(X_t^{|U_t|}) > F(X_t^{|U_t|-1}) > \cdots > F(X_t^0) = F(X_t)$.
\end{proof}
After the first iteration of Algorithm $2$, we get $X_1 = \{i \ |\  F(i|N-i) > 0\}$, and $Y_1 = \{i \ |\ F(i|\emptyset) \ge 0\}$. Based on Lemma $5$, we have $X_1 \subseteq Y_1$. Thus the reduced lattice $\mathcal{L} = [X_1,Y_1] \subseteq [\emptyset,N]$ is not empty, and we show that it contains all the global maxima.
\begin{lem}
Any global maximum of $F(X)$ is contained in the lattice $\mathcal{L} = [X_1,Y_1]$, i.e., $\forall X_* \in \argmax_{X \subseteq N}{F(X)}$, $X_1 \subseteq X_* \subseteq Y_1$.
\end{lem}
\begin{proof}
We prove $X_1 \subseteq X_*$. $X_* \subseteq Y_1$ can be proved in a similar way. Suppose $X_1 \not \subseteq X_*$, \emph{i.e.}, $\exists \ u \in X_1$, $u \not \in X_*$. According to the definition, $F(u|N-u) > 0 $. Since $X_* \subseteq N-u$, based on the SSBC property, we have $F(u|X_*) > 0$, that is $F(X_*+u) > F(X_*)$. This contradicts the optimality of $X_*$.
\end{proof}
At each iteration of Algorithm $2$, due to the SSBC property, $X_t$ only adds elements and $Y_t$ only removes elements. Thus we have $X_t \subseteq X_{t+1}$ and $Y_{t+1} \subseteq Y_t$, \emph{i.e.}, $\forall t$, $[X_{t+1},Y_{t+1}] \subseteq [X_t,Y_t]$. We denote the output lattice of Algorithm $2$ as $[X_+,Y_+]$. Then $[X_+,Y_+]$ is the smallest lattice in the chain which consists of the working lattices: $[X_+,Y_+] \subseteq \cdots \subseteq [X_t,Y_t] \subseteq \cdots \subseteq [X_1,Y_1] \subseteq [X_0,Y_0] = [\emptyset,N]$. Based on Lemma $5$, $[X_+,Y_+]$ is not empty, then we prove that it contains all the global maxima of $F$.

\begin{thm}
Suppose Algorithm $2$ outputs lattice $[X_+,Y_+]$. Any global maximum of $F(X)$ is contained in the lattice $\mathcal{L} = [X_+,Y_+]$, i.e., $\forall X_* \in \argmax_{X \subseteq N}{F(X)}$, $X_+ \subseteq X_* \subseteq Y_+$.
\end{thm}
\begin{proof}
Based on Lemma $7$, we have $X_1 \subseteq X_* \subseteq Y_1$. Suppose $X_t \subseteq X_* \subseteq Y_t$, we then prove $X_{t+1} \subseteq X_* \subseteq Y_{t+1}$. First, we suppose $X_{t+1} \not \subseteq X_*$. Because $X_{t+1} = X_t \cup U_t$, so $\exists \ u \in U_t$, $u \not \in X_*$. According to the definition of $U_t$, $F(u|Y_t-u) > 0$. Since $X_* \subseteq Y_t-u$, based on the SSBC property, we have $F(u|X_*) > 0$. This implies $F(X_*+u) > F(X_*)$, which contradicts the optimality of $X_*$. Hence $X_{t+1} \subseteq X_*$. And $X_* \subseteq Y_{t+1}$ can be proved in a similar way.
\end{proof}

Note that the proofs of Lemma $7$ and Theorem $2$ also work for local maximum cases, since we use singleton elements to construct contradictions.

\begin{lem}
Any local maximum of $F(X)$ is contained in the lattice $\mathcal{L} = [X_+,Y_+]$.
\end{lem}

Lemma $8$ indicates that if $X_+$ ($Y_+$) is a local maximum, it is the local maximum which contains the least (most) number of elements. Unfortunately, finding a local maximum for submodular functions is hard \cite{feige2011maximizing}, let alone quasi-submodular cases. Nonetheless, Algorithm $2$ provides an efficient strategy for search interval reduction, which is helpful because the reduction is on the exponential power. In the experimental section, we show the reduction can be quite surprising. Moreover, when an objective function has a unique local maximum, which is also the global maximum $X_+ = Y_+$, our algorithm can find it quickly.

\begin{thm}
Algorithm $2$ terminates in $\mathcal{O}(n)$ iterations. The time complexity is $\mathcal{O}(n^2)$.
\end{thm}
\begin{proof}
After each iteration, at least one element is removed from the current working lattice, so it takes $\mathcal{O}(n)$ iterations to terminate. At each iteration, all the elements in the current working lattice need to be queried once. Hence, the total complexity of Algorithm $2$ is $\mathcal{O}(n^2)$.
\end{proof}

\section{Discussions}

In Algorithm $1$, $Q_+$ and $S_+$ are local minima. While in Algorithm $2$, $X_+$ and $Y_+$ may not be local maxima. Is it possible to find a lattice for quasi-submodular maximization, where the endpoint sets are local maxima? We give an example to show that such a lattice may not exist. Suppose $N=\{1,2\}$, $F(\emptyset) = (N) = 1$, and $F(\{1\}) = F(\{2\})= 1.5$. It is easy to check that $F$ is submodular, thus quasi-submodular. The set of local maxima is $\{\{1\},\{2\}\}$. There is no local maximum which contains or is contained by all the other local maxima, since $\{1\}$ and $\{2\}$ are not comparable under the set inclusion relation.

As aforementioned, unconstrained quasi-submodular function maximization is NP-hard. While for unconstrained quasi-submodular function minimization, whether there exists a polynomial time algorithm or not is open now.

\section{Experimental Results}

In this section, we experimentally verify the effectiveness and efficiency of our proposed algorithms. We implement our algorithms using the SFO toolbox \cite{krause2010sfo} and Matlab. All experiments are run on a single core Intel i$5$ $2.8$ GHz CPU with $4$GB RAM.

We list several widely used quasi-submodular functions and the settings of our experiments as the following:

\begin{itemize}
\item Iwata's function $F(X) = |X||N \setminus X| - \sum\limits_{i \in X}{(5i-2n)}$ \cite{fujishige2011submodular}. The ground set cardinality is set to be $n = 5000$.
\item The COM (concave over modular) function $F(X) = \sqrt{w_1(X)} + w_2(N \setminus X)$, where $w_1$ and $w_2$ are randomly generated in $[0,1]^n$. This function is applied in speech corpora selection \cite{iyer2013fast}. The ground set cardinality is set to be $n = 5000$.
\item The half-products function $F(X) = \sum\limits_{i,j \in X, i \le j}{a(i)b(j)} - c(X)$, where $a, b, c$ are modular functions, and $a, b$ are non-negative. This function is employed in formulations of many scheduling problems and energy models \cite{boros2002pseudo}. Since $F$ is quasi-supermodular, we minimize $F$ through equivalently maximizing the quasi-submodular function $-F$, \emph{i.e.}, $\min{F} = -\max{(-F)}$. $n$ is set to be $100$.
\item The linearly perturbed functions. We consider the perturbed facility location function $F(X) = L(M,X) + \sigma (X)$, where $L(M,X)$ is the facility location function. $M$ is a $n \times d$ positive matrix. $\sigma$ is a $n$-dimensional modular function which denotes the perturbing noise of facility. We set $n = 100$, $d = 400$, and randomly generate $M$ in $[0.5,1]^{n \times d}$, the perturbing noise $\sigma$ in $[-0.01,0.01]^n$.
\item The determinant function $F(X) = det(K_X)$, where $K$ is a real $n \times n$ positive definite matrix indexed by the elements of $N$,
    and $K_X = [K_{ij}]_{i,j \in X}$ is the restriction of $K$ to the indices of $X$. This function is used to represent the
    sampling probability of determinantal point processes \cite{kulesza2012determinantal}. We set $n = 100$.
\item The multiplicatively separable function $F(X) = \Pi_{i=1}^k{F_i(X_i)}$. One example is the Cobb-Douglas production function $F(X) =
    \Pi_{i=1}^n{w(i)^{\alpha_i}}$, where $w \ge 0$ and $\alpha_i \ge 0$. This function is applied in economic fields
    \cite{topkis1998supermodularity}. We set $n = 2000$.
\end{itemize}

We are concerned with the approximation ratio of an optimization algorithm. We compare the approximation ratio and running time of UQSFMax with MMax \cite{iyer2013fast}. For MMax, we consider the following variants: random permutation (RP), randomized local search (RLS), and randomized bi-directional greedy (RG). For UQSFMax, we use it as the preprocessing steps of RP, RLS and RG, and denote the corresponding combined methods as URP, URLS, and URG.

For Iwata's function and COM function, $n = 5000$. In such an input scale, the exact branch-and-bound algorithm \cite{goldengorin1999data} cannot terminate because of the exponential time complexity. Actually, since the reduced lattices are quite small, we use the branch-and-bound method on the reduced lattices to obtain the exact optima.

\begin{table}[h]
\caption{Average lattice reduction rates.}
\begin{center}
\begin{tabular}[t]{|c|c|c|}
\hline
Algorithm & UQSFMax & UQSFMin\\
\hline
Iwata's function & 99.9\% & 99.9\%\\
\hline
COM function  & 99.5\% & 100.0\%\\
\hline
half-products function & 51.2\% & 48.8\%\\
\hline
linearly perturbed function & 99.3\% & 99.8\%\\
\hline
determinant function & 87.0\% & 72.6\%\\
\hline
Multiplicatively separable function & 100.0\% & 100.0\%\\
\hline
\end{tabular}
\end{center}
\end{table}

\begin{table*}[!t]
\caption{Approximation ratios of different algorithms and functions.}
\begin{center}
\begin{tabular}[t]{|c|c|c||c|c||c|c|}
\hline
Algorithm & RP & URP & RLS & URLS & RG & URG\\
\hline
Iwata's function & 0.94 & \bf{1.00} & 0.99 & \bf{1.00} & 0.98 & \bf{1.00}\\
\hline
COM function & 0.99 & \bf{1.00} & 0.99 & \bf{1.00} & 0.99 & \bf{1.00}\\
\hline
half-products function & 0.96 & \bf{0.97} & \bf{0.95} & 0.94 & 0.96 & \bf{0.99}\\
\hline
linearly perturbed function & 0.99 & \bf{1.00} & 0.99 & \bf{1.00} & 0.99 & \bf{1.00}\\
\hline
\end{tabular}
\end{center}
\end{table*}

\begin{table*}[!t]
\caption{Running time (seconds) of different algorithms and functions.}
\begin{center}
\begin{tabular}[t]{|c|c|c||c|c||c|c|}
\hline
Algorithm & RP & URP & RLS & URLS & RG & URG\\
\hline
Iwata's function & 96.18 & \bf{2.42} & 240.62 & \bf{2.47} & 194.30 & \bf{2.41}\\
\hline
COM function  & 43.85 & \bf{7.01} & 194.52 & \bf{6.91} & 366.43 & \bf{7.16}\\
\hline
half-products function & 0.35 & \bf{0.22} & 0.98 & \bf{0.52} & 9.96 & \bf{4.59}\\
\hline
linearly perturbed function & 1.37 & \bf{0.06} & 3.12 & \bf{0.06} & 15.92 & \bf{0.06}\\
\hline
\end{tabular}
\end{center}
\end{table*}

Table $2$ presents the approximation ratios while Table $3$ shows the running time. According to the comparison results, we find that using
our UQSFMax as the preprocessing steps of other approximation methods can reach comparable or better approximation performance while improve
the efficiency, since the UQSFMax can efficiently reduce the search spaces of other approximation algorithms, and the reduced
lattices definitely contain all the local and global optima as shown in the previous theoretical analysis.

Note that for non-submodular functions (determinant function and multiplicatively separable function), at present we have no efficient method to get approximate optima. So we cannot calculate the approximation ratios and we just record the average lattice reduction rates. We also record the rates of other functions for completeness.

The average lattice reduction rates are shown in Table $1$. This result also matches the running time. For example, the average lattice
reduction rate for half-products function is $51.2\%$, and the running time of URG is about a half of the running time of RG. For minimization, we have similar lattice reduction results, which are also presented in Table $1$.

\section{Related Work}

In this section, we introduce some related work of quasi-submodularity.

\subsection{Quasi-Supermodularity}

Quasi-supermodularity stems from economic fields. Milgrom and Shannon \cite{milgrom1994monotone} first propose the definition of
quasi-supermodularity. They find that the maximizer of a quasi-supermodular function is monotone as the parameter changes. In
combinatorial optimization, for quasi-submodular functions, this property means the set of minimizers has a nested structure,
which is the foundation of the proposed UQSFMin algorithm.

\begin{thm}[Reformulated from \cite{milgrom1994monotone}]
Given a quasi-submodular function $F: 2^N \mapsto \mathbb{R}$. $\forall A, B \subseteq N$, $A \subseteq B$, $\exists S_A \in \argmin_{S \subseteq A}{F(S)}$, $S_B \in \argmin_{S \subseteq B}{F(S)}$, s.t. $S_A \subseteq S_B$.
\end{thm}
\begin{proof}
Suppose $S_A \in \argmin_{S \subseteq A}{F(S)}$, $S_{B'} \in \argmin_{S \subseteq B}{F(S)}$. We have $F(S_A) \le F(S_A \cap S_{B'})$ because of $S_A \subseteq A$, $S_A \cap S_{B'} \subseteq A$ and $F(S_A) = \min_{S \subseteq A}{F(S)}$. According to quasi-submodularity, we have $F(S_A \cup S_{B'}) \le F(S_{B'})$. Denote $S_B \triangleq S_A \cup S_{B'}$. It is obvious that $S_B \in \argmin_{S \subseteq B}{F(S)}$ and $S_A \subseteq S_B$.
\end{proof}

Based on the theorem above, suppose we start from $X = \emptyset$, if $\exists i \in N \setminus X$, $F(X+i) < F(X)$, then we can set $X \leftarrow X+i$. This theorem ensures that there exists a chain structure of minimizers. This is a general principle. First, it works in submodular cases, for submodularity is a strict subset of quasi-submodularity. Moreover, when the superdifferential in \cite{iyer2013fast} is not superdifferential for non-submodular quasi-submodular functions, such as the determinant function and the multiplicatively separable functions, this principle can also hold.

In \cite{milgrom1994monotone}, only quasi-submodular function minimization (or equivalently, quasi-supermodular function maximization) is considered. For quasi-submodular function maximization, there is no existing study.

\subsection{Discrete Quasi-Convexity}

Another related direction is discrete quasi-convexity \cite{murota2003discrete,murota2003quasi}, which departs further from combinatorial optimization. In this paper, we consider set functions, \emph{i.e.}, functions defined on $\{0,1\}^n$. While in \cite{murota2003quasi}, quasi L-convex function, which is defined on $\mathbb{Z}^n$, is proposed.

In \cite{murota2003quasi}, quasi L-convex function is a kind of integer-valued function. When we restrict its domain from
$\mathbb{Z}^n$ to $\{0,1\}^n$, quasi L-convex function reduces to quasi-submodular function. Meanwhile, their results based on
$\mathbb{Z}^n$ domain extension reduces to trivial cases in combinatorial optimization. Hence, we view quasi L-convexity
\cite{murota2003quasi} as a generalization of quasi-submodularity based on domain extension, \emph{i.e.}, extending the domain from
$\{0,1\}^n$ to $\mathbb{Z}^n$.

\subsection{Submodularity}

As a special case of quasi-submodularity, submodularity should be the most related work to quasi-submodularity. Iyer et al. \cite{iyer2013fast} propose the  superdifferential based discrete Majorization-Minimization like algorithm, which performs lattice reduction for submodular function minimization. While the preliminary preservation algorithm \cite{goldengorin2009maximization} has the same effect for submodular function maximization.

As a weaker notion than submodularity, quasi-submodularity has no superdifferential, but it is also sufficient for lattice reduction. Thus the proposed UQSFMin algorithm can be viewed as a generalization of the MMin algorithm \cite{iyer2013fast}. One should note that since there is no known superdifferential for quasi-submodular function, our proof based on sub-single crossing property is quite different from the superdifferential based MMin algorithm. Generally, quasi-submodular function optimization is much harder than submodular function optimization.

Goldengorin \cite{goldengorin2009maximization} proposes the preliminary preservation algorithm (PPA), which is based on the preservation rules \cite{goldengorin1999maximization}. The preservation rule is another interpretation of the maximizers of submodular functions using set interval lattice partition. Unlike the superdifferential, we find that the preservation rules perfectly hold for not only submodular functions but also quasi-submodular functions. This provides an elegant principle for quasi-submodular function maximization. Using preservation rules for quasi-submodularity can also lead to the proposed UQSFMax algorithm. Thus we view UQSFMax as a generalization of PPA from submodular function maximization to quasi-submodular function maximization.

\subsection{Applications}

Unlike submodularity, quasi-submodularity is not well-known. Nonetheless, there are several applications related to quasi-submodularity scattered in different fields.

In rent seeking game, every contestant tends to maximize his probability of winning for a rent by adjusting his bidding. The payoff function of each contestant is quasi-submodular on his bidding and the total bidding of all the contestants (also called "aggregator"). Rent seeking game is a kind of aggregative quasi-submodular game, where each player's payoff function is quasi-submodular. We refer readers to \cite{schipper2004submodularity} for more details and examples of aggregative quasi-submodular games.

In minimum cut problems with parametric arc capacities, submodularity implies nested structural properties \cite{granot2012structural}.
While quasi-submodularity also leads to the same properties. But how to employ the properties to find an efficient max flow update algorithm for quasi-submodular functions is open at present \cite{granot2012structural}.

\section{Conclusions}

In this paper, we go beyond submodularity, and focus on a universal generalization of submodularity called quasi-submodularity. We propose two effective and efficient algorithms for unconstrained quasi-submodular function optimization. The theoretical analyses and experimental results demonstrate that although quasi-submodularity is a weaker property than submodularity, it has some good properties in optimization, which lead to lattice reduction while enable us to keep local and global optima in reduced lattices.

In our future work, we would like to make our algorithms exact for quasi-submodular function minimization and approximate for quasi-submodular function maximization if it is possible, and try to incorporate the constrained optimization into our framework.

{\small
\bibliographystyle{plain}
\bibliography{nipsbib}

\begin{thebibliography}{10}

\bibitem{boros2002pseudo}
E.~Boros and P.~L. Hammer.
\newblock Pseudo-boolean optimization.
\newblock {\em Discrete applied mathematics}, 123(1):155--225, 2002.

\bibitem{boyd2004convex}
S.~P. Boyd and L.~Vandenberghe.
\newblock {\em Convex optimization}.
\newblock Cambridge university press, 2004.

\bibitem{buchbinder2012tight}
N.~Buchbinder, M.~Feldman, J.~Naor, and R.~Schwartz.
\newblock A tight linear time (1/2)-approximation for unconstrained submodular
  maximization.
\newblock In {\em IEEE Annual Symposium on Foundations of Computer Science},
  pages 649--658, 2012.

\bibitem{das2011submodular}
A.~Das and D.~Kempe.
\newblock Submodular meets spectral: Greedy algorithms for subset selection,
  sparse approximation and dictionary selection.
\newblock In {\em International Conference on Machine Learning}, pages
  1057--1064, 2011.

\bibitem{edmonds2003submodular}
J.~Edmonds.
\newblock Submodular functions, matroids, and certain polyhedra.
\newblock In {\em Combinatorial Optimization-Eureka, You Shrink!}, pages
  11--26. Springer, 2003.

\bibitem{feige2009maximizing}
U.~Feige.
\newblock On maximizing welfare when utility functions are subadditive.
\newblock {\em SIAM Journal on Computing}, 39(1):122--142, 2009.

\bibitem{feige2011maximizing}
U.~Feige, V.~S. Mirrokni, and J.~Vondrak.
\newblock Maximizing non-monotone submodular functions.
\newblock {\em SIAM Journal on Computing}, 40(4):1133--1153, 2011.

\bibitem{fujishige2005submodular}
S.~Fujishige.
\newblock {\em Submodular functions and optimization}, volume~58.
\newblock Elsevier, 2005.

\bibitem{fujishige2011submodular}
S.~Fujishige and S.~Isotani.
\newblock A submodular function minimization algorithm based on the
  minimum-norm base.
\newblock {\em Pacific Journal of Optimization}, 7(1):3--17, 2011.

\bibitem{goldengorin2009maximization}
B.~Goldengorin.
\newblock Maximization of submodular functions: Theory and enumeration
  algorithms.
\newblock {\em European Journal of Operational Research}, 198(1):102--112,
  2009.

\bibitem{goldengorin1999data}
B.~Goldengorin, G.~Sierksma, G.~A. Tijssen, and M.~Tso.
\newblock The data-correcting algorithm for the minimization of supermodular
  functions.
\newblock {\em Management Science}, 45(11):1539--1551, 1999.

\bibitem{goldengorin1999maximization}
B.~Goldengorin, G.~A. Tijssen, and M.~Tso.
\newblock {\em The Maximization of Submodular Functions: Old and New Proofs for
  the Correctness of the Dichotomy Algorithm}.
\newblock University of Groningen, 1999.

\bibitem{golovin2011adaptive}
D.~Golovin and A.~Krause.
\newblock Adaptive submodularity: Theory and applications in active learning
  and stochastic optimization.
\newblock {\em Journal of Artificial Intelligence Research}, 42(1):427--486,
  2011.

\bibitem{granot2012structural}
F.~Granot, S.~T. McCormick, M.~Queyranne, and F.~Tardella.
\newblock Structural and algorithmic properties for parametric minimum cuts.
\newblock {\em Mathematical programming}, 135(1-2):337--367, 2012.

\bibitem{iyer2013fast}
R.~Iyer, S.~Jegelka, and J.~Bilmes.
\newblock Fast semidifferential-based submodular function optimization.
\newblock In {\em International Conference on Machine Learning}, pages
  855--863, 2013.

\bibitem{krause2010sfo}
A.~Krause.
\newblock Sfo: A toolbox for submodular function optimization.
\newblock {\em The Journal of Machine Learning Research}, 11:1141--1144, 2010.

\bibitem{krause2008near}
A.~Krause, A.~Singh, and C.~Guestrin.
\newblock Near-optimal sensor placements in gaussian processes: Theory,
  efficient algorithms and empirical studies.
\newblock {\em The Journal of Machine Learning Research}, 9:235--284, 2008.

\bibitem{kulesza2012determinantal}
A.~Kulesza and B.~Taskar.
\newblock Determinantal point processes for machine learning.
\newblock {\em Foundations and Trends in Machine Learning}, 5(2-3):123--286,
  2012.

\bibitem{leininger2006fending}
W.~Leininger.
\newblock Fending off one means fending off all: evolutionary stability in
  quasi-submodular aggregative games.
\newblock {\em Economic Theory}, 29(3):713--719, 2006.

\bibitem{liu2011entropy}
M.~Liu, O.~Tuzel, S.~Ramalingam, and R.~Chellappa.
\newblock Entropy rate superpixel segmentation.
\newblock In {\em IEEE Conference on Computer Vision and Pattern Recognition},
  pages 2097--2104, 2011.

\bibitem{lovasz1983submodular}
L.~Lov{\'a}sz.
\newblock Submodular functions and convexity.
\newblock In {\em Mathematical Programming The State of the Art}, pages
  235--257. Springer, 1983.

\bibitem{milgrom1994monotone}
P.~Milgrom and C.~Shannon.
\newblock Monotone comparative statics.
\newblock {\em Econometrica: Journal of the Econometric Society}, pages
  157--180, 1994.

\bibitem{murota2003discrete}
K.~Murota.
\newblock {\em Discrete convex analysis}, volume~10.
\newblock SIAM, 2003.

\bibitem{murota2003quasi}
K.~Murota and A.~Shioura.
\newblock Quasi m-convex and l-convex functions - quasiconvexity in discrete
  optimization.
\newblock {\em Discrete Applied Mathematics}, 131(2):467--494, 2003.

\bibitem{orlin2009faster}
J.~B. Orlin.
\newblock A faster strongly polynomial time algorithm for submodular function
  minimization.
\newblock {\em Mathematical Programming}, 118(2):237--251, 2009.

\bibitem{schipper2004submodularity}
B.~C. Schipper.
\newblock Submodularity and the evolution of walrasian behavior.
\newblock {\em International Journal of Game Theory}, 32(4):471--477, 2004.

\bibitem{singh2012bisubmodular}
A.~P. Singh, A.~Guillory, and J.~Bilmes.
\newblock On bisubmodular maximization.
\newblock In {\em International Conference on Artificial Intelligence and
  Statistics}, pages 1055--1063, 2012.

\bibitem{topkis1998supermodularity}
D.~M. Topkis.
\newblock {\em Supermodularity and complementarity}.
\newblock Princeton University Press, 1998.

\bibitem{vondrak2008optimal}
J.~Vondr{\'a}k.
\newblock Optimal approximation for the submodular welfare problem in the value
  oracle model.
\newblock In {\em ACM Symposium on Theory of Computing}, pages 67--74. ACM,
  2008.

\end{thebibliography}
}

\end{document}